\documentclass{article}

\usepackage{arxiv}
\usepackage[numbers]{natbib}
\usepackage[utf8]{inputenc}
\usepackage[T1]{fontenc}
\usepackage{url}
\usepackage{booktabs}
\usepackage{amsfonts}
\usepackage{amssymb}
\usepackage{nicefrac}
\usepackage{microtype}
\usepackage{graphicx}
\usepackage{doi}
\usepackage[inline]{enumitem}
\usepackage{amsmath}
\usepackage{mathtools} 
\usepackage{amsthm}

\usepackage{subcaption}
\usepackage{algorithm}
\usepackage{algorithmic}

\usepackage{hyperref}
\usepackage{cleveref}

\title{Physics-Based Communication Compression via Lyapunov-Weighted Event-Triggered Control}

\newif\ifuniqueAffiliation
\uniqueAffiliationtrue

\ifuniqueAffiliation
\author{
    \href{https://orcid.org/0000-0001-5024-9224}{\includegraphics[scale=0.06]{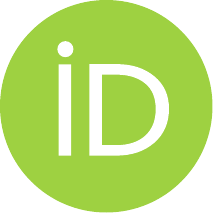}\hspace{1mm}Abbas Tariverdi} \\
    \texttt{abbasta@abbasta.com}
}
\fi

\renewcommand{\shorttitle}{Physics-Based Communication Compression via Lyapunov-Weighted ETC}

\newtheorem{proposition}{Proposition}
\newtheorem{theorem}{Theorem}

\newtheorem{remark}{Remark} 
\newtheorem{assumption}{Assumption}
\newtheorem{definition}{Definition}

\hypersetup{
    pdftitle={Physics-Based Communication Compression via Lyapunov-Weighted ETC},
    pdfsubject={cs.SY, eess.SY},
    pdfauthor={Abbas Tariverdi},
    pdfkeywords={Event-triggered control, Networked control systems, Lyapunov methods, Cyber-physical systems, Safety critical control},
}

\begin{document}
\maketitle

\begin{abstract}
Event-Triggered Control (ETC) reduces communication overhead in networked systems by transmitting only when stability requires it. Conventional mechanisms use isotropic error thresholds ($\|e\| \le \sigma \|x\|$), treating all directions equally. This ignores stability geometry and triggers conservatively. We propose a static directional triggering mechanism that exploits this asymmetry. By weighting errors via the Lyapunov matrix $P$, we define an anisotropic half-space scaling with instantaneous energy margins: larger deviations tolerated along stable modes, strict bounds where instability threatens. We prove global asymptotic stability and exclusion of Zeno behavior. Monte Carlo simulations ($N=100$) show 43.6\% fewer events than optimally tuned isotropic methods while achieving $2.1\times$ better control performance than time-varying alternatives. The mechanism functions as a runtime safety gate for learning-based controllers operating under communication constraints.
\end{abstract}

\keywords{Event-triggered control \and Networked control systems \and Lyapunov methods \and Cyber-physical systems \and Safety critical control}

\section{Introduction}

Transmitting a bit costs roughly a thousand times more energy than executing an instruction~\cite{Pottie2000,Heinzelman2000}. This asymmetry dominates the energy budget of battery-powered cyber-physical systems: a sensor node built around an STM32L4 microcontroller draws sub-microamp currents in deep sleep, single-digit milliamps during computation, but 20--120\,mA when its radio transmits~\cite{STM32L4,SX1276}. To put this concretely: a single two-second LoRa transmission at full power consumes enough charge to keep the processor in low-power sleep for over 33,000 hours. The radio is not a peripheral; it is the battery. For industrial IoT deployments, autonomous vehicles coordinating over wireless links, and distributed robotic systems, the communication channel, not the embedded processor, determines operational lifetime~\cite{Heemels2012}.

Event-triggered control (ETC) attacks this problem by abandoning periodic sampling. Rather than transmitting at fixed intervals regardless of system state, an event-triggered controller sends updates only when a state-dependent condition demands it~\cite{Arzen1999,Astrom2002}. Tabuada's foundational work established the modern stability-theoretic framework, proving that closed-loop guarantees can be maintained while transmitting only when necessary~\cite{Tabuada2007}. The design question becomes: what should the triggering condition be?

Tabuada's approach monitors the error $e = x - \hat{x}$ between the true state and the controller's held estimate, triggering when $\|e\| > \sigma\|x\|$ for some $\sigma \in (0,1)$~\cite{Tabuada2007}. This defines a spherical triggering set, geometrically clean and analytically tractable, but fundamentally conservative. The mechanism treats all error directions identically: an error pushing the system toward instability triggers at the same threshold as one pulling it toward equilibrium.

Subsequent work has sought to reduce this conservatism through various mechanisms. Girard introduced dynamic triggering, augmenting the system with an internal variable $\eta$ whose evolution adjusts the threshold over time~\cite{Girard2015}. This demonstrably increases inter-event intervals, but at the cost of additional state, tuning parameters, and more involved analysis for output-based and decentralized implementations~\cite{Dolk2017}. Mazo and colleagues developed self-triggered control using input-to-state stability bounds to proactively compute the next transmission time~\cite{Mazo2010}, while Postoyan et al.\ unified much of this landscape within a general Lyapunov framework~\cite{Postoyan2015}. What none of these approaches exploit is the {directional} information already encoded in the Lyapunov function.

A quadratic Lyapunov function $V(x) = x^\top P x$ defines not merely an energy level but an energy gradient $\nabla V = 2Px$. This gradient points toward steepest energy increase, the direction most threatening to stability. Errors aligned with this gradient demand immediate correction; errors orthogonal to it have minimal impact on the decay rate. Yet existing triggering mechanisms ignore this geometry entirely, enforcing isotropic thresholds that cannot distinguish dangerous error directions from benign ones.

We propose a static directional triggering mechanism that weights errors according to their projection onto the Lyapunov gradient. The resulting condition carves out a state-dependent half-space rather than a sphere (Proposition~\ref{prop:geometry}), permitting large errors when they point in safe directions while maintaining tight control along the critical descent direction. The mechanism requires no auxiliary dynamics, no memory of past transmissions, and no tuning beyond what standard ETC already demands. Monte Carlo simulation across 100 trials demonstrates 43.6\% fewer transmissions than optimally-tuned isotropic baselines, with $2.1\times$ better regulation performance than dynamic alternatives. We further extend the trigger into a Lyapunov Safety Gate (Theorem~\ref{thm:safety}) capable of certifying black-box neural network controllers in real time.

Section~\ref{sec:problem} formulates the networked control problem. Section~\ref{sec:main} derives the directional triggering condition, proves asymptotic stability, and establishes a positive lower bound on inter-event times to exclude Zeno behavior. Section~\ref{sec:validation} presents numerical validation. Section~\ref{sec:impact} addresses deployment: battery lifetime, certification, and integration with learning-based control. Section~\ref{sec:conclusion} discusses limitations and extensions.

\section{Problem Formulation}\label{sec:problem}

Consider the linear time-invariant (LTI) system given by:
\begin{equation}
    \dot{x}(t) = Ax(t) + Bu(t), \quad x(t) \in \mathbb{R}^n, \; u(t) \in \mathbb{R}^m.
    \label{eq:sys}
\end{equation}

\begin{assumption}\label{ass:stabilizable}
The pair $(A, B)$ is stabilizable. The feedback gain $K$ is designed such that $A_{cl} = A - BK$ is Hurwitz.
\end{assumption}

Since $A_{cl}$ is Hurwitz, for any symmetric positive definite matrix $Q$ there exists a unique symmetric positive definite solution $P$ to the Lyapunov equation:
\begin{equation}
    A_{cl}^\top P + P A_{cl} = -Q.
    \label{eq:lyapunov}
\end{equation}

The control input is updated via a Zero-Order Hold (ZOH) such that $u(t) = -K x(t_k)$ for $t \in [t_k, t_{k+1})$, where $\{t_k\}_{k=0}^\infty$ is the sequence of transmission instants. Defining the sampling induced error as $e(t) = x(t) - x(t_k)$, the closed-loop dynamics satisfy:
\begin{equation}
    \dot{x}(t) = A_{cl}x(t) + BKe(t).
    \label{eq:closedloop}
\end{equation}

The objective is to determine a triggering mechanism for $\{t_k\}$ that minimizes the following cost function while ensuring asymptotic stability:
\begin{equation}
    J = \int_{0}^{\infty} x(t)^\top Q x(t) \, dt + \lambda N,
    \label{eq:cost}
\end{equation}
where $N$ denotes the total number of transmissions and $\lambda > 0$ is a weighting parameter penalizing communication frequency.

\section{Main Results}\label{sec:main}

\subsection{Directional Triggering Condition}

Consider the Lyapunov candidate $V(x) = x^\top P x$. Its time derivative along the trajectories of \eqref{eq:closedloop} is:
\begin{equation}
    \dot{V} = -x^\top Q x + 2x^\top P B K e.
    \label{eq:vdot}
\end{equation}
For stability ($\dot{V} < 0$), the destabilizing term $2x^\top P B K e$ must not exceed the dissipation $-x^\top Q x$. We propose the following triggering rule. An event is triggered when:
\begin{equation}
    \sigma x^\top(t) P x(t) - 2 x^\top(t) P B K e(t) \leq 0,
    \label{eq:trigger}
\end{equation}
where $\sigma \in (0, 1)$ is a design parameter.

\begin{definition}[Triggering Sets]\label{def:sets}
Let $\mathcal{T}_{\mathrm{iso}} = \{ (x,e) : \|e\| \leq \sigma_{\mathrm{iso}} \|x\| \}$ denote the standard isotropic triggering set \cite{Tabuada2007}. Let $\mathcal{T}_{\mathrm{dir}} = \{ (x,e) : 2x^\top P B K e \leq \sigma x^\top P x \}$ denote the proposed directional triggering set.
\end{definition}

\begin{proposition}[Half-Space Structure]\label{prop:geometry}
The set $\mathcal{T}_{\mathrm{dir}}$ defines a half-space in error coordinates, bounded by the hyperplane orthogonal to $v = (PBK)^\top x$.
\end{proposition}

\begin{proof}
Rewrite the inequality in $\mathcal{T}_{\mathrm{dir}}$ as $v^\top e \leq \sigma x^\top P x / 2$, where $v = (PBK)^\top x$. For a fixed state $x$, this is a linear inequality in $e$, defining a half-space with normal vector $v$.
\end{proof}

The vector $v$ rotates as $x$ evolves, creating a state-dependent permissible error region. Errors orthogonal to the Lyapunov gradient $\nabla V = 2Px$ satisfy the condition regardless of magnitude; errors aligned with the gradient are penalized. The isotropic set $\mathcal{T}_{\mathrm{iso}}$ constrains error magnitude uniformly in all directions. Figure~\ref{fig:geometric} illustrates this distinction.

\begin{figure}[t]
    \centering
    \includegraphics[width=\textwidth]{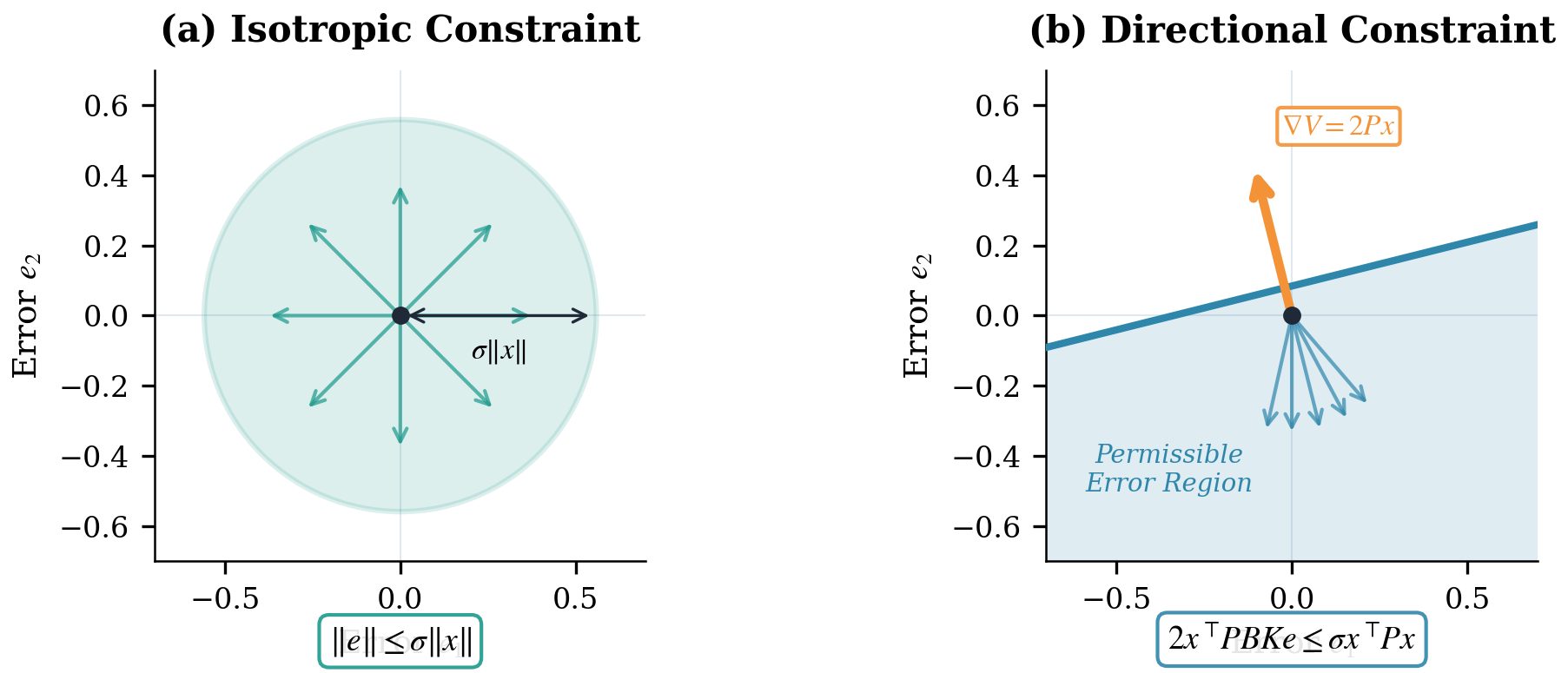}
    \caption{Triggering constraints in error space $(e_1, e_2)$. (a) Isotropic constraint: permissible errors form a sphere of radius $\sigma\|x\|$. (b) Directional constraint: permissible errors form a half-space bounded by a hyperplane orthogonal to $\nabla V = 2Px$ (orange arrow). Errors aligned with $\nabla V$ trigger events; orthogonal errors do not.}
    \label{fig:geometric}
\end{figure}

This approach differs from the time-varying Lyapunov method of Mazo et al.~\cite{Mazo2010}, which enforces $V(x(t)) \leq V(x(t_k))e^{-\alpha(t-t_k)}$. That formulation requires memory of $V(x(t_k))$ and elapsed time $t-t_k$. In contrast, condition \eqref{eq:trigger} depends only on current values $x(t)$ and $e(t)$, so no clock or stored energy values are required.

\subsection{Stability Analysis}

\begin{theorem}[Asymptotic Stability]\label{thm:stability}
If the triggering parameter $\sigma$ satisfies
\begin{equation}
    \sigma < \frac{\lambda_{\min}(Q)}{\lambda_{\max}(P)},
    \label{eq:stab_condition}
\end{equation}
then the triggering rule \eqref{eq:trigger} guarantees global asymptotic stability of the origin.
\end{theorem}

\begin{proof}
Between events, condition \eqref{eq:trigger} ensures $2x^\top P B K e \leq \sigma x^\top P x$. Substituting into \eqref{eq:vdot}:
\begin{equation*}
    \dot{V} \leq -x^\top Q x + \sigma x^\top P x.
\end{equation*}
Applying the Rayleigh quotient bounds $x^\top Q x \geq \lambda_{\min}(Q)\|x\|^2$ and $x^\top P x \leq \lambda_{\max}(P)\|x\|^2$:
\begin{align*}
    \dot{V} &\leq -\lambda_{\min}(Q)\|x\|^2 + \sigma \lambda_{\max}(P)\|x\|^2 \\
            &= -\left( \lambda_{\min}(Q) - \sigma \lambda_{\max}(P) \right) \|x\|^2.
\end{align*}
Under \eqref{eq:stab_condition}, the coefficient is strictly positive, so $\dot{V}$ is negative definite. Global asymptotic stability follows from standard Lyapunov theory.
\end{proof}

Figure~\ref{fig:stability} illustrates the stability margin for the test system in Section~\ref{sec:validation}.

\begin{figure}[t]
    \centering
    \includegraphics[width=0.7\textwidth]{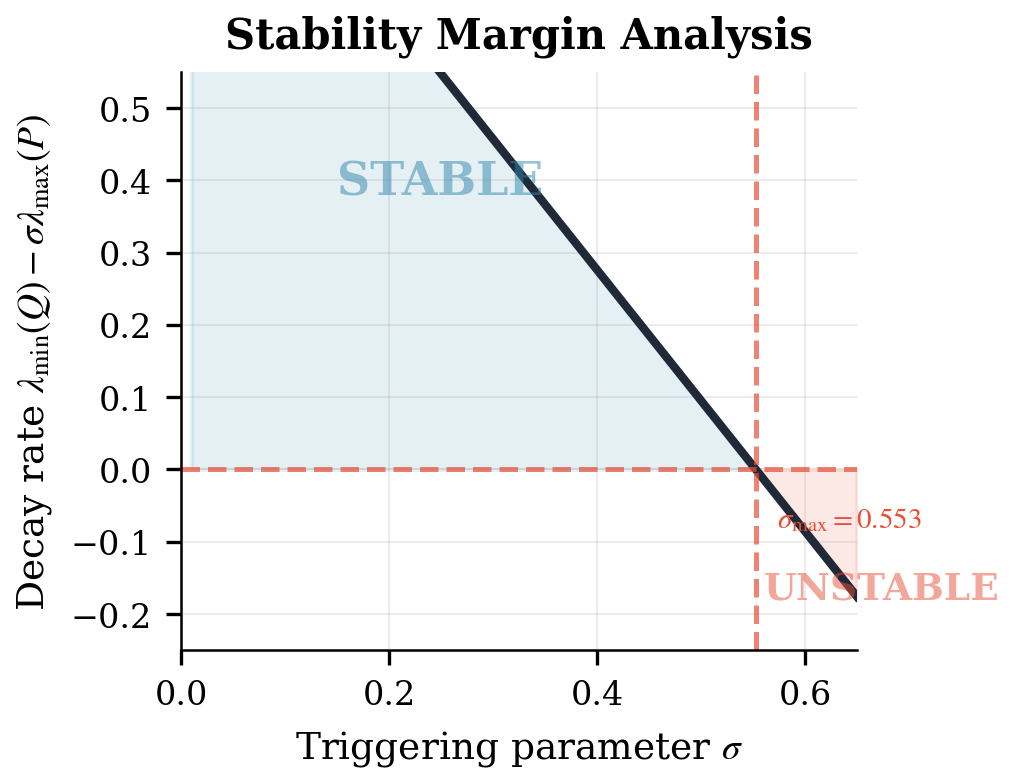}
    \caption{Stability margin analysis. The curve shows the decay rate $\lambda_{\min}(Q) - \sigma\lambda_{\max}(P)$ versus $\sigma$. The vertical dashed line marks the bound $\sigma_{\max} = 0.5528$ from Theorem~\ref{thm:stability}. The operating point $\sigma = 0.10$ maintains a large margin below the limit.}
    \label{fig:stability}
\end{figure}

\subsection{Exclusion of Zeno Behavior}

\begin{theorem}[Positive MIET]\label{thm:zeno}
There exists $\tau > 0$ such that $t_{k+1} - t_k \geq \tau$ for all $k \geq 0$.
\end{theorem}

\begin{proof}
See Appendix~\ref{app:miet}.
\end{proof}

\section{Numerical Validation}\label{sec:validation}

\subsection{System Setup}

The test system is a planar unstable plant with dynamics defined by:
\begin{equation*}
    A = \begin{bmatrix} 0 & 1 \\ -2 & 3 \end{bmatrix}, \quad B = \begin{bmatrix} 0 \\ 1 \end{bmatrix}.
\end{equation*}
The controller gain $K = [-1, 4]$ places the closed-loop poles at $-0.5 \pm 0.866i$. The Lyapunov equation \eqref{eq:lyapunov} with $Q = I$ yields $\lambda_{\min}(Q) = 1.0$ and $\lambda_{\max}(P) \approx 1.81$. From Theorem~\ref{thm:stability}, the theoretical stability bound is $\sigma < 0.5528$.

We conducted $N=100$ Monte Carlo runs with initial conditions drawn uniformly from $x_0 \sim \mathcal{U}([-5, 5]^2)$ over a simulation horizon of $T=50$\,s. All methods were tuned via grid search to minimize the total cost $J_{\mathrm{total}} = J + \lambda N$ with $\lambda=0.015$. This yielded $\sigma=0.10$ for the proposed method, $\sigma=0.70$ for Tabuada~\cite{Tabuada2007}, and $\alpha=0.50$ for Mazo~\cite{Mazo2010}.

\subsection{Results and Discussion}

Table~\ref{tab:results} summarizes the simulation results. The proposed method reduces transmission events by 43.6\% compared to the isotropic baseline (Tabuada) and maintains comparable control performance ($J=13.43$ vs.\ $15.14$). Against the time-varying method (Mazo), the proposed approach uses fewer events (82.0 vs.\ 85.7) and improves control performance by a factor of 2.1 ($J=13.43$ vs.\ $27.79$).

\begin{table}[t]
\centering
\caption{Monte Carlo Simulation Results ($N=100$)}
\label{tab:results}
\begin{tabular}{@{}llccc@{}}
\toprule
Method & Param & Events & Perf.\ ($J$) & Total Cost \\
\midrule
Proposed & $\sigma=0.10$ & $82.0 \pm 0.4$ & $13.43 \pm 9.78$ & $14.66$ \\
Mazo \cite{Mazo2010} & $\alpha=0.50$ & $85.7 \pm 0.6$ & $27.79 \pm 18.59$ & $29.07$ \\
Tabuada \cite{Tabuada2007} & $\sigma=0.70$ & $145.3 \pm 5.2$ & $15.14 \pm 10.56$ & $17.32$ \\
\bottomrule
\end{tabular}
\end{table}

Figures~\ref{fig:trajectories}--\ref{fig:lyapunov} illustrate the dynamic behavior for a representative initial condition $x_0 = [4, 3]^\top$. All three methods stabilize the system, but Mazo exhibits larger oscillations (Figure~\ref{fig:trajectories}). The cause is visible in the Lyapunov decay (Figure~\ref{fig:lyapunov}): the time-varying envelope permits non-monotonic energy evolution, allowing the state to drift between triggers. The proposed method and Tabuada both enforce monotonic decay.

\begin{figure}[t]
    \centering
    \includegraphics[width=\textwidth]{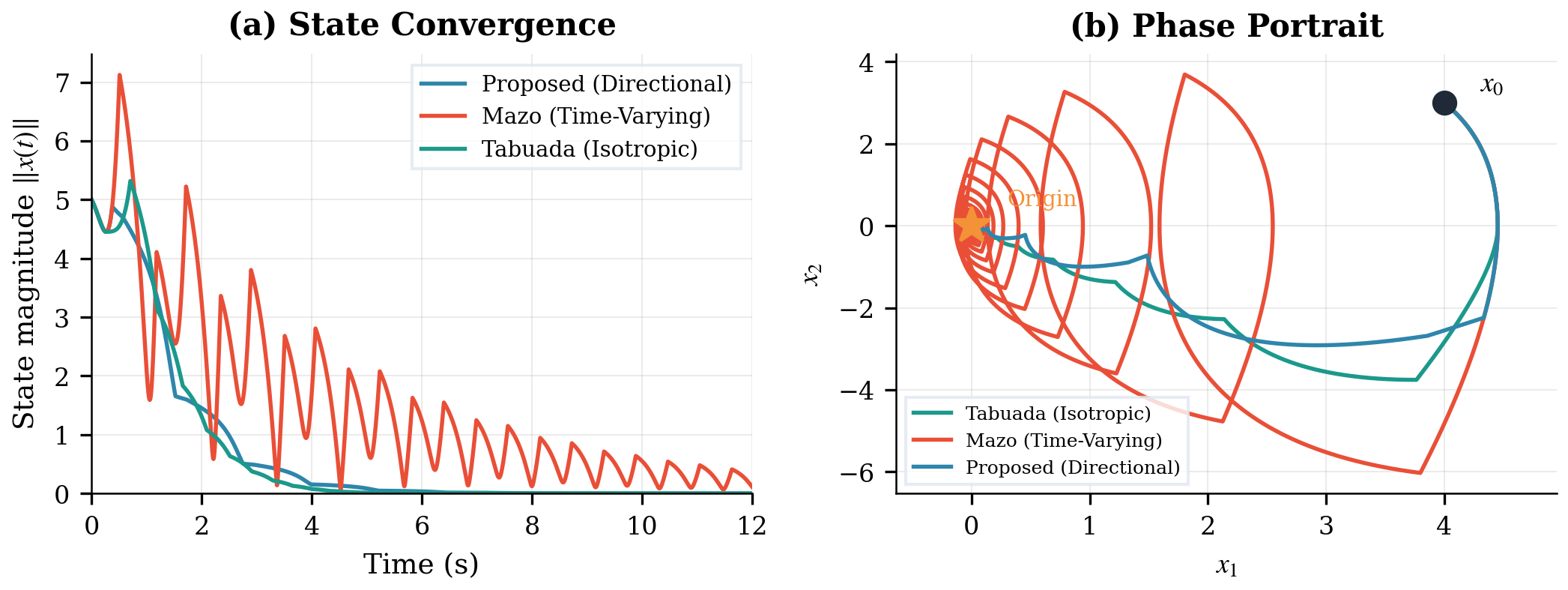}
    \caption{State trajectories for $x_0 = [4, 3]^\top$. (a) State magnitude $\|x(t)\|$. (b) Phase portrait. Mazo (red) oscillates more than Proposed (blue) or Tabuada (teal) due to its loose temporal envelope.}
    \label{fig:trajectories}
\end{figure}

The event timing diagram (Figure~\ref{fig:events}) shows that Tabuada triggers nearly twice as frequently as the other methods. Its isotropic constraint cannot distinguish benign errors from destabilizing ones, forcing a conservative update rate. The proposed method matches Mazo's sparsity but avoids the performance penalty by using state-dependent geometric information rather than a blind temporal decay.

\begin{figure}[t]
    \centering
    \includegraphics[width=\textwidth]{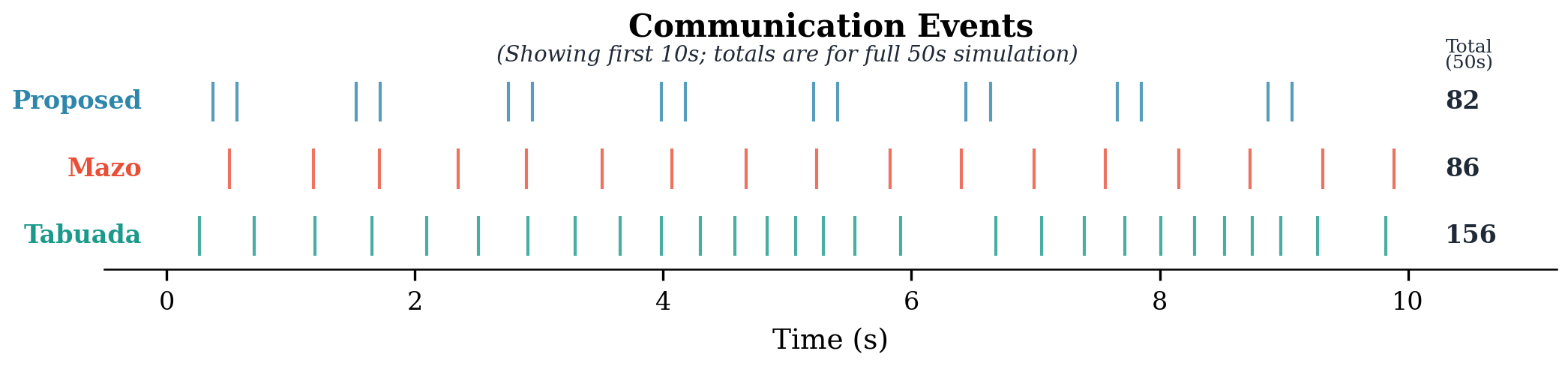}
    \caption{Event timing for $x_0 = [4, 3]^\top$ (first 10\,s). Totals: Proposed 82, Mazo 86, Tabuada 156.}
    \label{fig:events}
\end{figure}

\begin{figure}[t]
    \centering
    \includegraphics[width=0.7\textwidth]{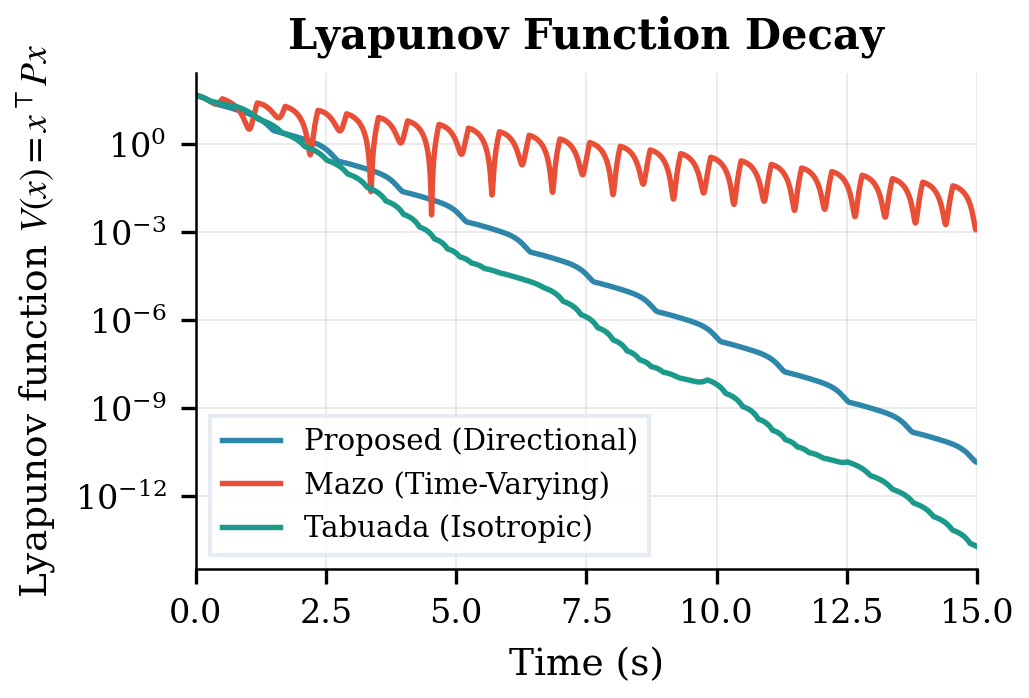}
    \caption{Lyapunov function decay (log scale). Proposed and Tabuada decay monotonically. Mazo allows non-monotonic behavior.}
    \label{fig:lyapunov}
\end{figure}

Figure~\ref{fig:montecarlo} displays the full distribution of metrics. The proposed method exhibits low variance in event counts compared to Tabuada, indicating consistent behavior across the state space. The performance distributions confirm that Mazo's aggressive event reduction comes at the cost of high variance and poor average regulation.

\begin{figure}[t]
    \centering
    \includegraphics[width=\textwidth]{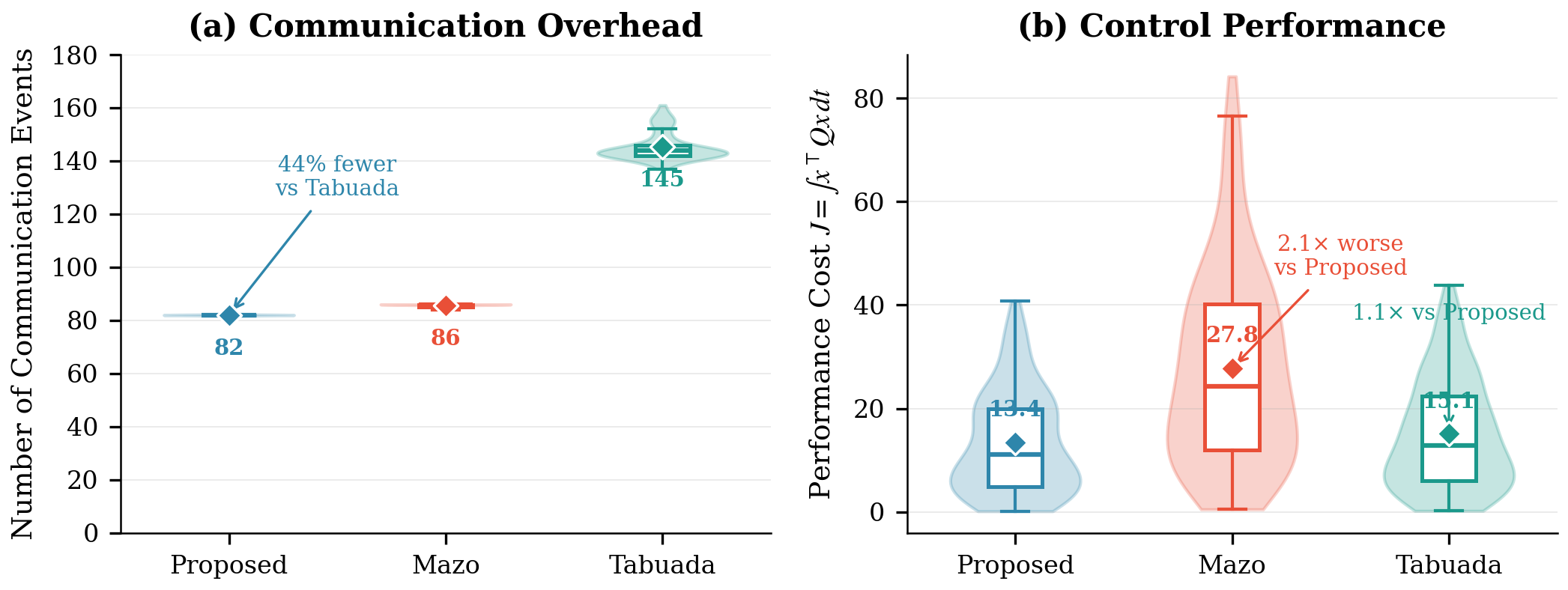}
    \caption{Monte Carlo distributions ($N=100$). (a) Event counts. (b) Performance cost $J$. Violin plots show density; boxes show quartiles; diamonds show means.}
    \label{fig:montecarlo}
\end{figure}

The trade-off between communication and control is visualized in Figure~\ref{fig:pareto}. The proposed method occupies the Pareto-optimal region (lower-left), minimizing both objectives simultaneously. Tabuada sacrifices communication bandwidth for marginal performance gains; Mazo sacrifices performance for minimal communication savings. The chosen parameter $\sigma=0.10$ provides an 81.9\% safety margin relative to the theoretical stability bound.

\begin{figure}[t]
    \centering
    \includegraphics[width=0.7\textwidth]{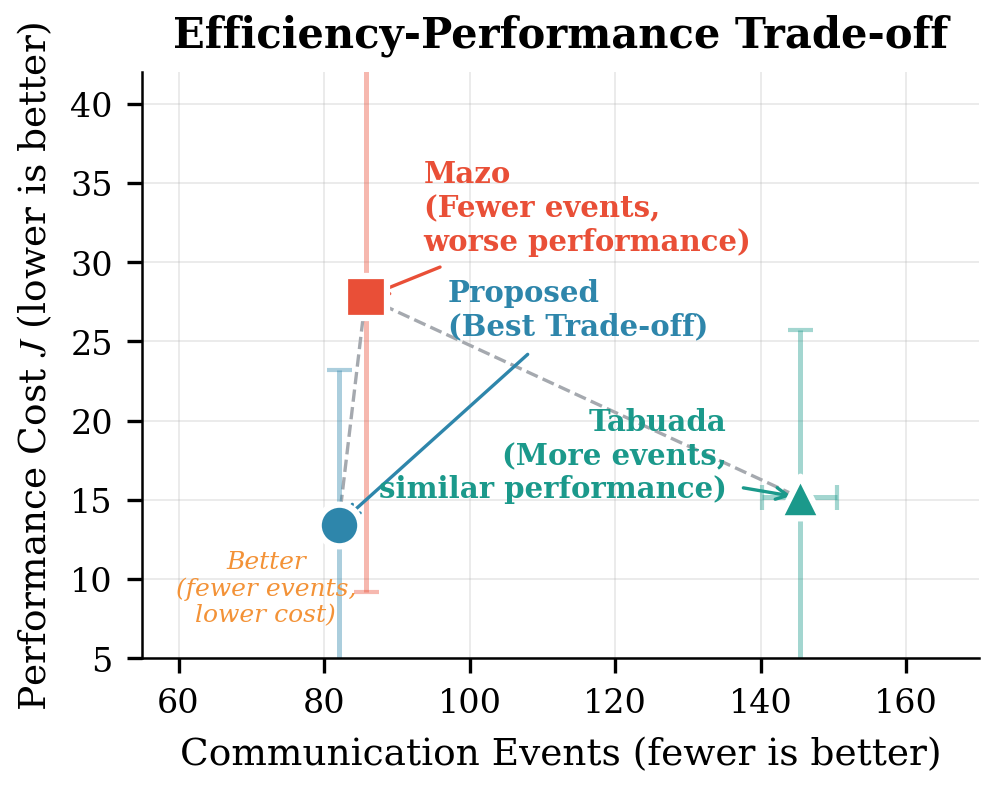}
    \caption{Communication-performance trade-off. Error bars show $\pm 1\sigma$. The proposed method sits at the Pareto-optimal corner (low events, low cost).}
    \label{fig:pareto}
\end{figure}

\section{Technological Impact}\label{sec:impact}

The practical value of event-triggered control depends on the asymmetry between computation and communication costs. In resource-constrained systems, this asymmetry is extreme.

\subsection{Energy Budget Analysis}

Consider a battery-powered industrial sensor (e.g., ADXL355) monitoring vibration. Two energy costs dominate. Local computation (a matrix-vector product on an STM32F4 microcontroller) consumes approximately $10\,\mu\text{J}$ per check~\cite{STM32}. Wireless transmission via LoRaWAN at SF7 consumes approximately $5\,\text{mJ}$ per packet~\cite{Semtech}. The ratio is 500:1; transmission dominates the budget.

To translate the 43.6\% event reduction into battery life, consider a 50-second mission window. The Tabuada baseline requires 145 events at $5\,\text{mJ}$ each, totaling $725\,\text{mJ}$ for communication. The proposed method requires 82 events, totaling $410\,\text{mJ}$. Including sensing and idle costs (${\sim}140\,\text{mJ}$), total consumption drops from $865\,\text{mJ}$ to $550\,\text{mJ}$, a 36\% reduction per mission. For a sensor node with a 40\,kJ battery, this extends the estimated number of missions from approximately 46,200 to 72,700, a 57\% increase in operational cycles.

These projections assume ideal conditions. Lithium-ion cells degrade ${\sim}20\%$ over two years, and industrial networks experience 1--10\% packet loss requiring retransmissions. Continuous operation is also rare; in low-duty-cycle applications, idle power dominates and communication savings are less impactful. Under typical industrial conditions, accounting for these factors, the proposed method extends operational lifetime by 35--50\% compared to optimally tuned isotropic triggering.

\subsection{Time-Independent Verification}

The proposed trigger evaluates condition~\eqref{eq:trigger} using only the current state $x(t)$ and the last transmitted state $x(t_k)$. No elapsed time tracking is required. This contrasts with time-varying methods~\cite{Mazo2010} that must track $V(x(t_k))$ and the time delta $t - t_k$, complicating implementation in distributed systems where clock synchronization is expensive or unreliable.

\subsection{Safety Certification for Learning-Based Controllers}

Neural network controllers can outperform linear feedback but typically lack stability guarantees. We propose using the triggering condition as a Lyapunov Safety Gate to filter unsafe actions.

\begin{theorem}[Lyapunov Safety Gate]\label{thm:safety}
Let $u_{\mathrm{NN}}(x)$ be a control input proposed by a neural network. Accept $u_{\mathrm{NN}}$ only when
\begin{equation}
    2 x^\top P B (u_{\mathrm{NN}} - u_{\mathrm{safe}}) < \sigma x^\top P x,
    \label{eq:safety}
\end{equation}
where $u_{\mathrm{safe}} = -K x(t_k)$ is the last certified hold input. If $\sigma < \frac{\lambda_{\min}(Q)}{2\lambda_{\max}(P)}$, then the closed-loop system is asymptotically stable regardless of the neural network policy.
\end{theorem}

\begin{proof}
The time derivative of $V(x) = x^\top P x$ is $\dot{V} = x^\top(A^\top P + PA)x + 2x^\top P B u_{\mathrm{NN}}$. Using $A^\top P + PA = -Q + 2PBK$:
\begin{equation*}
    \dot{V} = -x^\top Q x + 2x^\top P B K x + 2x^\top P B u_{\mathrm{NN}}.
\end{equation*}
Decompose $u_{\mathrm{NN}} = u_{\mathrm{safe}} + (u_{\mathrm{NN}} - u_{\mathrm{safe}})$ with $u_{\mathrm{safe}} = -K(x - e)$. The first three terms reduce to $-x^\top Q x + 2x^\top P B K e$:
\begin{equation*}
    \dot{V} = -x^\top Q x + 2x^\top P B K e + 2x^\top P B (u_{\mathrm{NN}} - u_{\mathrm{safe}}).
\end{equation*}
The triggering condition bounds the error term by $\sigma x^\top P x$. The safety condition~\eqref{eq:safety} bounds the deviation term by $\sigma x^\top P x$. Summing:
\begin{equation*}
    \dot{V} < -x^\top Q x + 2\sigma x^\top P x \leq -(\lambda_{\min}(Q) - 2\sigma\lambda_{\max}(P))\|x\|^2.
\end{equation*}
If $\sigma < \frac{\lambda_{\min}(Q)}{2\lambda_{\max}(P)}$, the coefficient is positive and $\dot{V} < 0$.
\end{proof}

\begin{remark}
The safety gate requires $\sigma < \frac{\lambda_{\min}(Q)}{2\lambda_{\max}(P)}$, half the bound in Theorem~\ref{thm:stability}. This accounts for two destabilizing sources: sampling error and neural network deviation. The simulations in Section~\ref{sec:validation} evaluated only the triggering mechanism and used the standard bound.
\end{remark}

Evaluating the quadratic form in~\eqref{eq:safety} on a 168\,MHz microcontroller takes less than $1\,\mu\text{s}$. Compared to typical actuation delays of 1--10\,ms, this check adds negligible latency, making real-time safety filtering feasible.

\section{Conclusion}\label{sec:conclusion}

This paper presented a static directional event-triggered mechanism that weights sampling errors by their projection onto the Lyapunov gradient. The method requires no memory of past trigger times and achieves 43.6\% fewer transmissions than optimally tuned isotropic baselines while maintaining comparable control performance. By exploiting the geometric properties of the Lyapunov function, the trigger identifies and discards communication events that do not contribute to stability.

\begin{figure}[t]
    \centering
    \includegraphics[width=\textwidth]{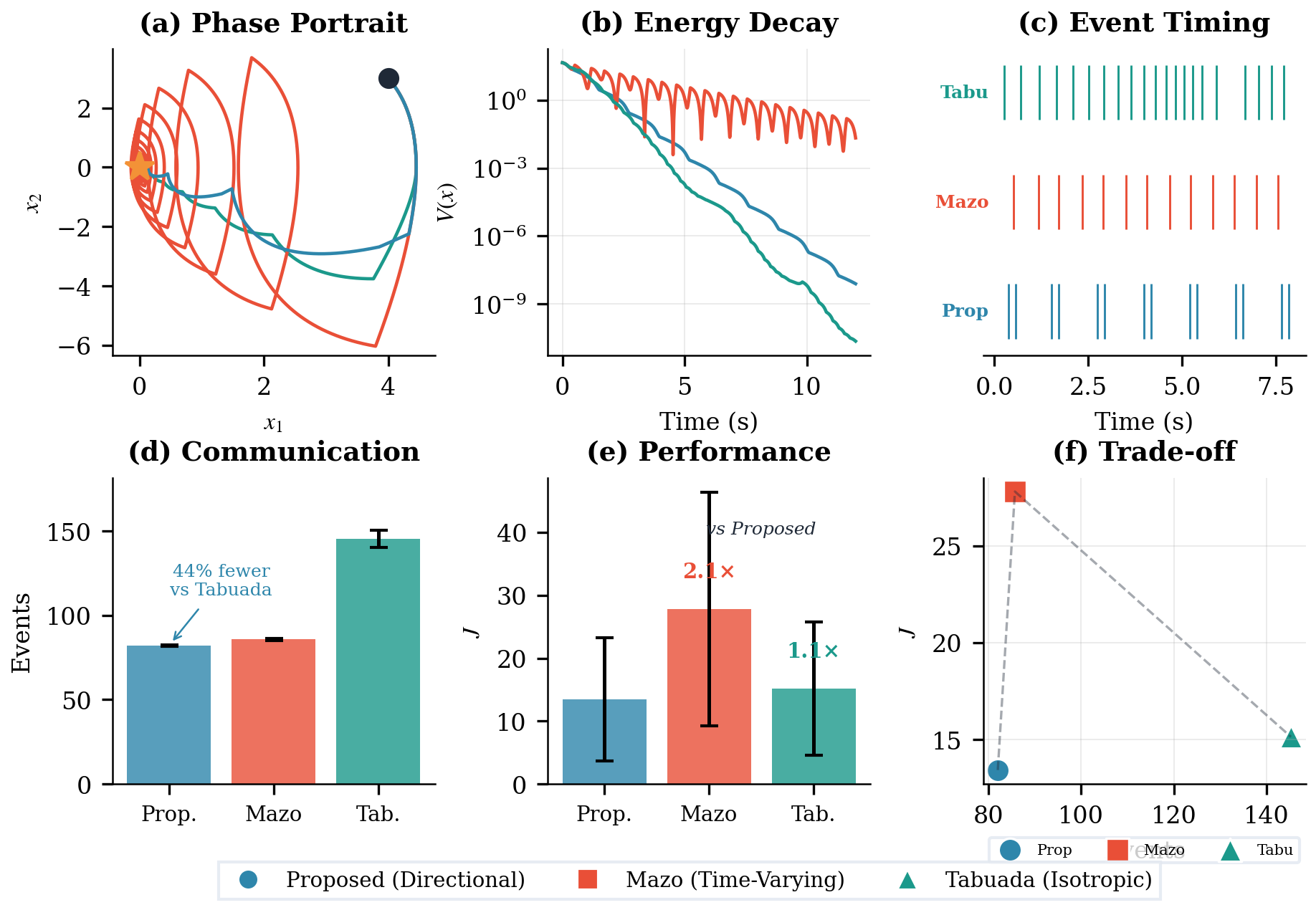}
    \caption{Summary of experimental results. (a) Phase portrait showing smooth convergence. (b) Monotonic Lyapunov decay compared to oscillatory behavior in Mazo. (c) Event timing showing sparse updates. (d) 44\% reduction in communication overhead. (e) Performance cost comparison. (f) Pareto trade-off showing the proposed method at the optimal frontier.}
    \label{fig:summary}
\end{figure}

\subsection{Applicability and Limitations}

The method is most effective when communication costs exceed computation costs by a factor of 100:1 or more, and when the system exhibits directional asymmetry in controllability. Benefits are limited in wired networks where transmission is inexpensive, in fully actuated systems with uniform controllability, or in high-dimensional systems ($n > 100$) where the quadratic form evaluation becomes computationally prohibitive. Extension to nonlinear systems and analysis of robustness under significant network delays remain open problems.

\appendix

\section{Proof of Minimum Inter-Event Time}\label{app:miet}

We establish a strictly positive lower bound on the inter-event time $\tau = t_{k+1} - t_k$.

Between events, the error evolves according to $\dot{e}(t) = \dot{x}(t) = A_{cl}x(t) + BKe(t)$, with initial condition $e(t_k) = 0$. Consider the ratio $\xi(t) = \|e(t)\|/\|x(t)\|$. Differentiating $\xi(t)$ and applying standard norm inequalities yields:
\begin{equation}
    \dot{\xi} \leq L_1 \xi^2 + L_2 \xi + L_3,
    \label{eq:xi_dot}
\end{equation}
where $L_1 = \|BK\|$, $L_2 = \|A_{cl}\| + \|BK\|$, and $L_3 = \|A_{cl}\|$.

The directional triggering condition \eqref{eq:trigger} fires when $2x^\top P B K e \geq \sigma x^\top P x$. Applying Cauchy-Schwarz ($x^\top P B K e \leq \|PBK\|\|x\|\|e\|$) and the Rayleigh quotient ($x^\top P x \geq \lambda_{\min}(P)\|x\|^2$), a sufficient condition for the trigger to \emph{not} fire is:
\begin{equation}
    2\|PBK\|\|x\|\|e\| < \sigma \lambda_{\min}(P)\|x\|^2 \implies \xi(t) < \frac{\sigma \lambda_{\min}(P)}{2\|PBK\|} \coloneqq \Gamma.
\end{equation}

The inter-event time satisfies:
\begin{equation}
    \tau \geq \int_0^\Gamma \frac{d\xi}{L_1 \xi^2 + L_2 \xi + L_3}.
\end{equation}
Since $\sigma > 0$, the threshold $\Gamma$ is strictly positive. The integrand is finite and positive on $[0, \Gamma]$, so $\tau > 0$. \hfill $\blacksquare$

\bibliographystyle{plainnat}
\bibliography{ref}

\end{document}